\documentclass[journal,final,twoside]{IEEEtran}
\usepackage[T1]{fontenc}
\usepackage{graphicx}
\usepackage{cite}
\usepackage{amsmath,amssymb,amsfonts}
\usepackage{algorithmic}
\usepackage{graphicx}
\usepackage{textcomp}
\usepackage{xcolor}
\usepackage{amsthm}
\usepackage{tikz}
\usepackage{color,soul}
\usetikzlibrary{svg.path}
\usepackage{scalerel}
\usepackage{cuted}
\newcommand{\SINR}{\mathrm{SINR}}

\newcommand{\rmin}{r_{\mathrm{min}}}
\newcommand{\rmax}{r_{\mathrm{max}}}
\newcommand{\rEarth}{r_{\oplus}}
\newcommand{\phiuser}{\phi_{\mathrm{u}}}
\newcommand{\phisat}{\phi_{\mathrm{s}}}
\newcommand{\Phisat}{\Phi_{\mathrm{s}}}
\newcommand{\lamsat}{\lambda_{\mathrm{s}}}

\newcommand{\Insp}{I_{}}

\newcommand{\thmin}{\theta_{\mathrm{min}}}

\newcommand{\Prob}{{\mathbb{P}}}

\newcommand{\Pc}{P_\mathrm{c}}

\newcommand{\ps}{p_\mathrm{0}}
\newcommand{\Gs}{G_\mathrm{0}}
\newcommand{\Rs}{R_\mathrm{0}}

\newcommand{\rs}{r_\mathrm{0}}
\newcommand{\ths}{\theta_\mathrm{s}}
\newcommand{\Xs}{\mathcal{X}_\mathrm{0}}
\newcommand{\xs}{\mathit{x}_\mathrm{0}}

\newcommand{\Xn}{\mathcal{X}_n}
\newcommand{\mus}{\mu_\mathrm{s}}
\newcommand{\devs}{\sigma_\mathrm{s}}

\newcommand{\ppi}{p_n}
\newcommand{\mn}{m_n}
\newcommand{\phiu}{\phi_\mathrm{u}}

\newcommand{\NI}{N_\mathrm{I}}
\newcommand{\Nvis}{N_\mathrm{vis}}

\newtheorem{theorem}{Theorem}
\newtheorem{lem}{Lemma}
\newtheorem{cor}{Corollary}

\definecolor{orcidlogocol}{HTML}{A6CE39}
\tikzset{
  orcidlogo/.pic={
    \fill[orcidlogocol] svg{M256,128c0,70.7-57.3,128-128,128C57.3,256,0,198.7,0,128C0,57.3,57.3,0,128,0C198.7,0,256,57.3,256,128z};
    \fill[white] svg{M86.3,186.2H70.9V79.1h15.4v48.4V186.2z}
                 svg{M108.9,79.1h41.6c39.6,0,57,28.3,57,53.6c0,27.5-21.5,53.6-56.8,53.6h-41.8V79.1z M124.3,172.4h24.5c34.9,0,42.9-26.5,42.9-39.7c0-21.5-13.7-39.7-43.7-39.7h-23.7V172.4z}
                 svg{M88.7,56.8c0,5.5-4.5,10.1-10.1,10.1c-5.6,0-10.1-4.6-10.1-10.1c0-5.6,4.5-10.1,10.1-10.1C84.2,46.7,88.7,51.3,88.7,56.8z};
  }
}

\newcommand\orcidicon[1]{\href{https://orcid.org/#1}{\mbox{\scalerel*{
\begin{tikzpicture}[yscale=-1,transform shape]
\pic{orcidlogo};
\end{tikzpicture}
}{|}}}}

\usepackage{hyperref} 
\title{Nonhomogeneous Stochastic Geometry Analysis of Massive LEO Communication Constellations}
\author{\IEEEauthorblockN{Niloofar Okati \orcidicon{0000-0002-8074-5146} and Taneli Riihonen \orcidicon{0000-0001-5416-5263},~\IEEEmembership{Member,~IEEE} }%
\thanks{The work of N.~Okati and T.~Riihonen was supported by a Nokia University Donation.}%
\thanks{N.~Okati and T.~Riihonen are with Unit of Electrical Engineering, Faculty of Information Technology and Communication Sciences, Tampere University, FI-33720 Tampere, Finland (e-mail: niloofar.okati@tuni.fi; taneli.riihonen@tuni.fi).}
\thanks{}%
}
\begin{document}
\setlength{\parskip}{0pt}

\maketitle

\begin{abstract}
Providing truly ubiquitous connectivity 
requires development of low Earth orbit (LEO) satellite Internet, whose theoretical study is lagging behind network-specific simulations.
In this paper, we derive analytical expressions for the downlink coverage probability and average data rate of a massive inclined LEO constellation in terms of {total interference power's Laplace transform in the} presence of fading and shadowing, ergo presenting a stochastic geometry-based analysis. We assume the desired link to experience Nakagami-$m$ fading, which serves to represent different fading scenarios by varying integer $m$, while the interfering channels can follow any fading model without an
effect on analytical tractability. To take into account the inherent {non-uniform distribution of satellites across different latitudes, we} model the LEO network as a nonhomogeneous Poisson point process with its intensity being a function of satellites’ actual distribution in terms of {constellation size, the altitude of the constellation, }and the inclination of orbital planes. 
From the numerical results, we observe optimum points for both the constellation altitude and the number of orthogonal frequency channels; interestingly, the optimum user's latitude is greater than the inclination angle due to the presence of fewer interfering satellites.
Overall, the presented study facilitates general stochastic evaluation and {planning of satellite Internet constellations} without specific orbital simulations or tracking data on satellites' exact positions in space.
\end{abstract}


\IEEEpeerreviewmaketitle
\section{Introduction}
\label{sec:intro}

Recent advances towards 6th generation (6G) wireless networks require progression and development {of non-terrestrial networks to provide seamless connections with high transmission capacity}\cite{GiordaniNon-Terrestrial,JiaJointHAP,NTN6G3,NTN6G4}. Among non-terrestrial networks, low Earth orbit (LEO) satellite Internet constellations have gained increasing popularity as they provide global connectivity for unserved or underserved regions, where the deployment of terrestrial networks is not feasible or economically reasonable\cite{46,47}. Deploying thousands of satellites will ensure that every single person or appliance on Earth could be connected and no location is left in outage.   

While the performance of many LEO constellations (e.g., Starlink, OneWeb, Kuiper) has been evaluated through network-specific simulations to put the commercial plans forward, a general scientific understanding of their performance is limited in the open literature. Conventional simulation-based studies are restricted to few number of satellites with deterministic locations which is not capable of evaluating the general performance of a massive satellite network consisting of thousands of satellites. Moreover, in most of the literature, the coverage regions are assumed to have fixed circular footprints, while selecting smaller inclination angles and {simultaneous consolidated operation of several LEO networks render a not-so-regular Voronoi tessellation.}

In this paper, downlink coverage probability and average data rate of inclined LEO constellations are analyzed under general shadowing and fading propagation models. The satellites' positions are assumed to be distributed as a nonhomogeneous Poisson point process (NPPP), which models {the satellites distribution across varying latitudes} precisely by setting the intensity function to be the actual density of satellites in an actual constellation. 

\subsection{Related Works}

The literature around LEO networks is mostly limited to deterministic and simulation-based analyses. In \cite{1}, {the performance of two different LEO constellations was simulated assuming specific constellation sizes.}
{The probability of average call drop and the distribution of the number of handoffs were studied for the Iridium constellation} in \cite{4}. 
{A deterministic model to characterize the visibility time of one LEO satellite was presented} in \cite{8}. { Since the model in} \cite{8} {is not valid for any arbitrarily located user, authors in} \cite{49} {contributed statistical analysis of coverage time in a mobile LEO constellation}.
In \cite{48}, a LEO-based Internet-of-Things architecture was presented so as to supply network access for devices distributed in remote areas.

Stochastic geometry is an area of mathematics, which deals with the study of random objects on Euclidean space. In the area of telecommunication, stochastic geometry has been extensively utilized {to model, evaluate, and develop the wireless communication networks with irrigular topologies} \cite{24,blaszczyszyn2018stochastic,25}, especially for two-dimensional (planar) terrestrial networks  \cite{23,24,25,9,10,11,12,28,blaszczyszyn2018stochastic}. {Various studies in stochastic geometry modeling of multi-tier and cognitive networks were reviewed} in \cite{23}. {Observations in} \cite{9} {have shown that the Poisson point process (PPP) and a regular grid model provide lower and upper bounds on the network performance metrics, respectively, with the same deviation from the actual network performance.}
The research in \cite{10, 11}{, being an extension to} \cite{9}, {modeled a multi-tier network considering the limitations for the achievable quality-of-service}. {The coverage of uplink was studied in} \cite{12} {assuming base stations and devices are distributed as independent PPPs.}

The application of stochastic geometry to three-dimensional wireless networks has gained {remarkable} attention in the literature \cite{19,20}.
In \cite{19}, a PPP model was applied to model and analyze the coverage in three-dimensional cellular networks.
{Since the PPP, despite providing tractable analysis, {is not accurate when applied to networks with limited nodes in a finite area} \cite{41}, a binomial point process (BPP) can be utilized instead to capture the characteristics of such networks \cite{42,43}. {A finite network of unmanned aerial vehicles was modeled as a BPP in} \cite{20,45}. {In} \cite{42,44},  {a planar network with an arbitrary shape was studied assuming the transmitter is positioned at a fixed distance. The results were then generalized in} \cite{28}, {by using two protocols for selecting the transmitter}.}

On the literature around the LEO satellite networks, the analysis is limited to few number of satellites with known locations and/or coverage spots. In \cite{Afshang}, {with tools from stochastic geometry, authors have developed a method to characterize the magnitude of Doppler shifts in a LEO network. Resource control of a satellite--terrestrial network was investigated in} \cite{51}, {in order to minimize the outage probability and maximize the data rate.
Focusing on only a single spotbeam, the hybrid satellite--terrestrial network supporting 5G infrastructure has been presented in}~\cite{50,52}.
In \cite{dwivedi2020performance}, the outage probability of a satellite-based Internet-of-Things, in which LEO satellites relay uplink data to ground stations, is derived in closed form by assuming a low number of satellites at deterministic locations.

Recently, more research on LEO networks using stochastic geometry has started emerging. In our study\cite{okati2020downlink}, {generic performance of satellite networking has been formulated assuming uniform distribution for satellites, and without considering any explicit model of orbits and constellation geometry. However, in practical constellations, the satellites are not evenly distributed across different latitudes} \cite{okatiPIMRC}, i.e., as the user gets farther from the equator towards the poles, more satellites are visible to it. {Therefore, the constellation densities are typically not uniform in practice which results in an inconsistency between the performance of actual and uniformly distributed constellations.} In \cite{okati2020downlink}, we compensated for this mismatch by numerically computing a parameter named as effective number of satellites, while in \cite{okatiPIMRC}, we derived a mathematical expression for it based on constellation geometry.

The work in \cite{AlouiniNearestNeighbor} characterizes the distance distribution in two different communication links in a LEO satellite network: link between a user on Earth and the nearest satellite to it and the link between a satellite and its nearest neighboring satellite. Unlike in \cite{okati2020downlink}, the satellites are assumed to be placed at different altitudes, i.e., on multiple orbital shells. Stochastic geometry and the results from \cite{AlouiniNearestNeighbor} were then utilized in \cite{talgat2020stochastic} to {obtain the downlink probability of coverage for a LEO network}, where satellite gateways act as relays between the satellites and users on Earth.  An uplink communication scenario was characterized by considering interfering terrestrial transmitters in \cite{YastrebovaGlobcom}. In \cite{okatiVTC2021}, the coverage and rate of a noise-limited {\em interference-free} LEO network were analyzed assuming the satellites' positions are distributed as a nonhomogeneous Poisson point process (NPPP), which models the actual distribution of satellites along different latitudes more precisely by selecting proper intensity for it.

\subsection{Contributions and Paper Organization}

We model the satellites' positions in a LEO network as a nonhomogeneous PPP which facilitates not only using the tools from stochastic geometry, but also capturing the exact characteristics of the actual constellations, i.e., the uneven distribution of satellites across different latitudes. Unlike in\cite{okati2020downlink,okatiPIMRC,talgat2020stochastic}, by selecting the intensity of NPPP to fit the actual distribution of satellites on an orbital shell, there is no mismatch between the performance of theoretical stochastic constellations and actual deterministic LEO networks.
We derive the intensity of NPPP in closed form in terms of the constellation parameters: the total number of satellites, altitude of the constellation, latitude of the satellites, and the inclination of the orbits.

{As the main contributions, we utilize stochastic geometry to formulate the coverage and average achievable rate} of a user served by a LEO constellation in terms of the derivative of Laplace transform of interference power.\footnote{Thus, the present study, unlike the preliminary results presented in \cite{okatiVTC2021} that are limited to the special case of scheduling an orthogonal channel for every satellite, includes the cumulative interference from all other satellites that are visible to the user and share the same channel.} Our derivations do not rely on exact location of every single satellite and are applicable for performance analysis of any given constellation as long as the constellation parameters are known. Modeling the satellites' locality as a NPPP, {the analytical expressions obtained from a stochastic constellation geometry can be particularly used to analyze the actual deterministic constellations.} 

In propagation modelling, unlike most of literature on this topic, we take into account the effect of shadowing caused by the presence of the obstacles surrounding the user. 
To retain analytical tractability and still cover different fading scenarios, we assume Nakagami-$m$ fading with integer $m$ as well as shadowing with any desired distribution for the propagation model of desired links.\footnote{Varying the value of $m$, we are able to control the multi-path fading severity. For instance, $m=1$ corresponds to Rayleigh fading environment while $m\to\infty$ represents non-fading channels.} {For interfering signals, any arbitrarily distributed fading and shadowing can be considered, since the analytical tractability is unaffected.}

Finally, we evaluate two critical performance metrics, i.e., coverage and data rate, in terms of several key design parameters, such as altitude of the constellation and the number of frequency channels. From the numerical results, we are able to observe optimal points for these parameters for some specific network setup. Counter-intuitively, the user which resides in higher latitudes, away from the constellation borders, has the best performance due to existence of fewer interfering satellites in that region.

{The remainder of this paper is organized as follows.} Section~\ref{sec:sysmod} describes the system model and the mathematical preliminaries for modeling a LEO network as a NPPP. {The main outcome of this study, which is the derivation of analytical expressions for downlink coverage probability and average achievable rate of a terrestrial user, is presented in} Section~\ref{sec:cov}, which involves also the analysis of the Laplace transform of interference power. We provide numerical results in Section~\ref{sec:Numerical Results} for the verification of our derivations and studying the effect of key system parameters such as the size of the constellation and its altitude as well as the channel parameters on the network performance. Finally, we conclude the paper in Section~\ref{sec:Conclusion}.

\section{System Model}
\label{sec:sysmod}
In this section, first, we present the characteristics and geometries of actual low Earth orbit satellite constellations. Next, we will introduce the mathematical preliminaries for modeling the actual network as a stochastic point process.

\subsection{Actual Inclined Constellations}

{As shown in }Fig.~\ref{fig:system_model}, {we consider a LEO communication satellite constellation consisting of $N$ satellites launched uniformly on circular orbits with inclination angle, $\iota$, and altitude that is denoted by $\rmin$ --- the subscript indicates the minimum possible distance between a satellite and a ground user (as measured at the zenith).} Satellites' spherical coordinates in terms of their latitude and longitude are denoted by $(\phisat,\lamsat)$.

A user terminal is located on any specific latitude, denoted by $\phiu$, on the surface of Earth that is approximated as a perfect sphere with radius $\rEarth \approx 6371$~km. {Satellites rising above the horizon at an angle of $\theta_\mathrm{s}\geq \thmin$ are the only ones capable of transmitting signals to the users. As such, $\rmax$ refers to the maximum distance at which a satellite and a user are able to communicate (and it occurs when $\theta_\mathrm{s}=\thmin$), and}
\begin{align}
\frac{\rmax}{\rEarth} = \sqrt{\frac{\rmin}{\rEarth}\left(\frac{\rmin}{\rEarth}+2\right) + \sin^2(\thmin)}-\sin(\thmin).
\end{align}

\begin{figure}[t]
    \includegraphics[trim =8cm 6cm 5cm 5.5cm, clip,width=0.5\textwidth]{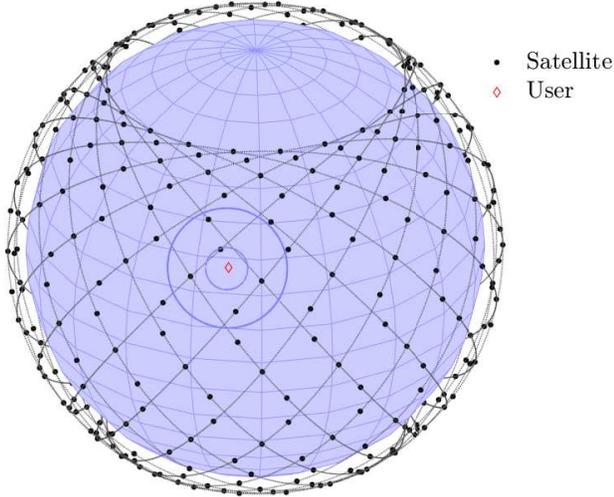}
    \caption{A constellation in an example case of $N = 400$ satellites flying on $\iota = 53^\circ$ inclined orbits. The borders of two spherical caps above the user are shown: the outer one covers all visible satellites to the user while the inner one is empty of satellites and the serving satellite is located on its border.}
    \label{fig:system_model}
\end{figure}

In this paper, the serving satellite is the one with the shortest distance to the user. {Assuming that $K$, with $K\leq N$, orthogonal frequency channels are available for the network, we will distribute $N/K$ satellites randomly among the channels, which potentially causes co-channel interference to the user. All satellites on the same channel that are elevated above the horizon to an angle of $\theta_\mathrm{s}\geq \thmin$ cause interference to reception of the user}.

{The variables $R_0$ and $R_n$, \mbox{$n=1,2,\ldots,N$}, represent the distances from the user to the serving satellite and the other interfering satellites, respectively, while $G_0$ and $G_n$ denote the corresponding channel gains to model fading.}
Shadowing effect is modeled by random variables $\Xn$,  \mbox{$n=1,2,\ldots,N$}, correspondingly. Obviously, $G_n=\Xn=0$ if $R_n > \rmax$ for some $n=1,\ldots,N$. {To simplify notation}, {when $\NI>0$, we let indices $n=1,2,\ldots,\NI$ correspond to those $\NI \leq N/K-1$ satellites with  $\theta_\mathrm{s}\geq \thmin$ that cause co-channel interference.} {To take into account the effect of directed transmission, we assign serving and interfering satellites to different power levels that are denoted by $\ps$ and $\ppi$, \mbox{$n=1,2,\ldots,N$}, respectively, such that $\ppi\leq\ps$.} Table~\ref{table:math notations} {summarizes the notation used in this paper.}

{According to the described model,} {the signal-to-interference-plus-noise ratio (SINR)} {of the link is given by}
\begin{align}
\label{eq:SINR}
\SINR = \frac{\ps \Gs \Xs \Rs ^{-\alpha}}{I+\sigma^2},
\end{align}
{where the constant $\sigma^2$ is the additive noise power}, {the parameter $\alpha$ is a path loss exponent,} and
\begin{align}
    \label{I_NSP}
   I \triangleq \sum_{n=1}^{\NI} \ppi G_n \Xn R_n^{-\alpha}
\end{align}
{is the cumulative interference power from all $\NI$ other satellites above the user's horizon} {that share the same frequency channel with the serving satellite.} The distance from the user to its nearest satellite is 
\begin{align}
    \label{NSP condition}
   \Rs=\min_{n=1,2,\ldots,\Nvis} R_n,
\end{align}
where $\Nvis$ is a variable representing the number of visible satellites to the user (cf.\ the outer cap in Fig.~\ref{fig:system_model}).

\begin{table*}[!t]

\caption{Summary of mathematical notation}
\label{table:math notations}
\begin{center}
\begin{tabular}{c|cc}
\hline
\hline
\textbf{Notation}&\textbf{Description}\\
\hline
$\rEarth; \rmin; \rmax $&Earth radius ($6371$ km); {Constellation altitude}; Maximum possible distance between a user and a visible satellite\\
\hline
$\Rs$; $R_n$ &Serving distance; Distance to the $n^{\mathrm{th}}$ satellite\\
\hline
$N; \NI; K$& Constellation size; The number of interfering satellites; The number of frequency channels\\

\hline
$\Gs; G_n$&Channel fading gain of the serving link; Channel fading gain of the $n^{\mathrm{th}}$ link\\
\hline

$\Xs; \Xn$& Shadowing component of the serving link; Shadowing component of the $n^{\mathrm{th}}$ link\\
\hline
 $\alpha$&Path loss exponent\\
 \hline
$\ps; \ppi $& Transmission power of the serving satellite; Transmission power of the interfering satellites\\
\hline
 $\sigma^2$&Additive noise power\\
\hline
 $T$&SINR threshold for coverage probability\\
\hline
 $\Pc; \bar{C}$&Coverage probability; Average achievable rate\\
\hline
\hline
\end{tabular}
\end{center}
\end{table*}

\subsection{Nonhomogeneous PPP Model}
\label{sec:performance}

In the constellation described earlier, the satellites {appear unevenly along the lines of latitudes}, which means, e.g., that there are more visible satellites for a user located close to inclination limits than for one on equatorial region. In order to model the latitude-dependent distribution of satellites, we assume that the satellites are distributed according to a {\em nonhomogeneous} PPP, $\xi$, on a {spherical surface} with radius $\rEarth+\rmin$. The NPPP is characterized with its intensity, $\delta(\phisat,\lamsat)$, which varies according to the satellites' latitude (and/or longitude).

By the definition of a NPPP, the number of points in some bounded region ${\cal A}$ of the orbital shell is a Poisson-distributed random variable denoted by $\mathcal{N}$. 
Thereby, the probability to have $n$ satellites in ${\cal A}$ is given by
\begin{align}
\label{void prob}
&P_n\left({\cal A}\right)\triangleq\Prob\left(\mathcal{N}=n\right)\\\nonumber
&= \frac{1}{n!}\left(\iint_{\cal A}\delta(\phisat,\lamsat) \,(\rmin+\rEarth)^2 \cos(\phisat)\, d\phisat d\lamsat\right)^n\\\nonumber
&\times\exp\left(-\iint_{\cal A}\delta(\phisat,\lamsat) \,(\rmin+\rEarth)^2 \cos(\phisat)\, d\phisat d\lamsat\right),
\end{align}
where $\delta(\phisat,\lamsat)$ is the intensity function of nonhomogeneous PPP at latitude $\phisat$ and longitude $\lamsat$. Based on the given system model, ${\cal A}$ is the spherical cap where {viewable satellites to the user exist} (cf.\ the outer one in Fig.~\ref{fig:system_model}), with surface area $\left({\delta\pi \left(\rmax^2-\rs^2\right)}\right)/({1-\frac{\rmin}{\rEarth+\rmin}})$ (See \cite[Appendix~A]{okati2020downlink}).

In order to precisely model a LEO network as a NPPP, we first need to characterize the intensity function based on the actual physical network as follows. The preliminaries obtained herein will be used shortly towards contributing expressions for probability of coverage and average achievable rate.

\begin{lem}
\label{lem:Intensity ppp}
When satellites are distributed uniformly on low circular orbits with the altitude, $\rmin$, and the inclination angle, $\iota$, the intensity function of the nonhomogeneous PPP is a function of latitude, $\phisat$, only and given by
\begin{equation}
    \label{eq:lambda}
   \delta(\phisat)=\frac{N}{\sqrt{2}\pi^2(\rmin+\rEarth)^2}\cdot\frac{1}{ \sqrt{\cos(2\phisat)-\cos(2\iota)}},
\end{equation}
and we can denote $\delta(\phisat,\lamsat) = \delta(\phisat)$ since it does not depend on $\lamsat$, for $\phisat\in[-\iota,\iota]$.
\end{lem}

\begin{proof}
The intensity function is equivalent to the actual density of the satellites on an orbital shell element created by spanning the azimuthal angle from $0^{\circ}$ to $360^{\circ}$ on the orbital shell at latitude $\phisat$,
{that is calculated as}
\begin{align}
\label{eq:delta2}
\delta(\phisat) = \frac{N f_{\Phisat}(\phisat)\,d\phisat}{2\pi(\rmin+\rEarth)^2\cos(\phisat)\,d\phisat},
\end{align}
{which is the ratio of the number of satellites resided on the surface element to the element's area.} Substituting the probability density  $f_{\Phisat}(\phisat)$ of random latitude $\Phisat$ \cite[Lemma~2]{okatiPIMRC} completes the proof.
\end{proof}

If the intensity of satellites is simplistically presumed to be uniform all over the orbital shell, it can be written as follows.

\begin{lem}
\label{lem:constant Intensity}
When satellites are uniformly distributed on a sphere, the point process turns into a homogeneous Poisson point process with a constant intensity given by 
\begin{equation}
    \label{eq: constant lambda}
   \delta=\frac{N}{4\pi (\rmin+\rEarth)^2},
\end{equation}
which does not depend on latitudinal/longitudinal parameters. 
\end{lem}

Thus, for the special case when satellites are distributed uniformly on the orbital shell, by substitution from Lemma~\ref{lem:constant Intensity}, the probability given in \eqref{void prob} can be expressed in closed form as
\begin{align}
    \label{void_cnst}
    P_n\left({\cal A}\right)
   = \frac{1}{n!}\hspace{-1 mm}\left(\frac{N \left(\rmax^2-\rmin^2\right)}{4\rEarth\left(\rEarth+\rmin\right)}\right)^n  \hspace{-2 mm}\exp \left(-\frac{N \left(\rmax^2-\rmin^2\right)}{4\rEarth\left(\rEarth+\rmin\right)}\right).
\end{align}


\section{Performance Analysis}
\label{sec:cov}

In this section, we focus on the performance analysis of a LEO satellite network in terms of coverage probability and data rate of a user in an arbitrary location on Earth. We utilize stochastic geometry in order to formulate coverage probability and rate as a function of the network and the propagation parameters. Two main components of our analytical derivations are the distribution of the distance to the nearest satellite and the Laplace function of interference which will be presented throughout this section.

\subsection{The Distance to The Nearest Satellite}
In this paper, {the serving satellite is assumed to be the nearest one to the user.}
Therefore, an important parameter for coverage and rate analysis is the probability density function (PDF) of the distance to the nearest satellite, $\Rs$, which is given as follows.

\begin{lem}
\label{Lem:pdf_r0}
The PDF of the serving distance $\Rs$, when the satellites are distributed according to a nonhomogeneous PPP with a latitude-dependent intensity, $\delta(\phisat)$, is 
\begin{align}
\label{eq:pdf_r0}
\nonumber
&f_{\Rs}\left(\rs\right)\\\nonumber
 &= 2\rs\left(\frac{\rmin}{\rEarth}+1\right)\exp(-\gamma(\rs))\int_{\max(\phiu-\ths,-\iota)}^{\min(\phiu+\ths,\iota)}\delta(\phisat)\\
 &\hspace{3 cm}\times\frac{\cos(\phisat)}{\sqrt{\cos^2(\phisat-\phiu)-\cos^2(\ths)}}d\phisat,
\end{align}
where 
\begin{align}
\label{eq:gamma}
\nonumber
&\gamma(\rs)=2(\rmin+\rEarth)^2\\
&\times\int_{\max(\phiu-\ths,-\iota)}^{\min(\phiu+\ths,\iota)}\delta(\phisat)\cos(\phisat)\cos^{-1}\left({\frac{\cos(\ths)}{\cos(\phisat-\phiu)}}\right) d\phisat
\end{align}
and $\ths$ is the polar angle difference between the serving satellite and the user which is given by $\ths=\cos^{-1}\left(1-\frac{\rs^2-\rmin^2}{2(\rmin+\rEarth)\rEarth}\right)$. Equation~\eqref{eq:pdf_r0} is valid for $\ths\geq |{\phiu}| - \iota$ and $\rs \in [\rmin, 2\rEarth+\rmin]$ while $f_{\Rs}\left(\rs\right)=0$ otherwise.
\end{lem}
\begin{proof}
See Appendix~\ref{Appen:proof of Lemma 2}.
\end{proof}

When the density of satellites is presumed to be uniform, i.e., it is not a function of latitude, the PDF of the serving distance can be obtained in closed form as follows. 

\begin{lem}
\label{Lem:r0 uniform}
The PDF of the serving distance $\Rs$ when the satellites are distributed uniformly with constant intensity given in Lemma~\ref{lem:constant Intensity} is 
\begin{align}
\label{eq:pdf_r0_uni}
f_{\Rs}\left(\rs\right)=\frac{N\rs}{2\rEarth(\rmin+\rEarth)}\exp\left(-N\left(\frac{\rs^2-\rmin^2}{4(\rmin+\rEarth)\rEarth}\right)\right)
\end{align}
for $\rs \in [\rmin, 2\rEarth+\rmin]$ while $f_{\Rs}\left(\rs\right)=0$ otherwise.
\end{lem}

\begin{proof}
{In this proof, the same principles are used as in} Lemma~\ref{Lem:pdf_r0}. However, the integration from a constant density over the cap sphere will reduce to a simple expression. Thus, $ F_{\Rs}\left(\rs\right)=1-\exp\left(-N\left(\frac{\rs^2-\rmin^2}{4(\rmin+\rEarth)\rEarth}\right)\right)$. Calculating the derivative of the CDF w.r.t.\ $\rs$ completes the proof.
\end{proof}

The PDF of the serving distance is also derived in \cite{okati2020downlink} assuming satellites are uniformly distributed as a BPP. It is worth noting that the Taylor series expansion of Lemma~\ref{Lem:r0 uniform} and the serving distance distribution given in \cite[Lemma~2]{okati2020downlink} are the same for the first two terms. The difference between the serving distance in a uniformly distributed constellation given in \cite{okati2020downlink} and the homogeneous Poisson point process in Lemma~\ref{Lem:r0 uniform} is insignificant since the argument of exponential function in \eqref{eq:pdf_r0_uni}, i.e., $N\left(\frac{\rs^2-\rmin^2}{4(\rmin+\rEarth)\rEarth}\right)$, is small.

\subsection{Coverage Probability and Average Data Rate}

The coverage probability is the probability that the SINR at the receiver is higher than the {minimum SINR required to successfully transmit the data. The coverage probability is defined as}
\begin{align}
\label{eq:coverage probability}
\Pc\left(T\right) \triangleq \Prob\left(\SINR > T\right)=\Prob\left(\frac{\ps \Gs \Xs \Rs^{-\alpha}}{\Insp+\sigma^2}>T\right),
\end{align}
where $T$ is the SINR threshold.

Using the above definition, we express the coverage probability of a user in the following theorem. 

\begin{theorem}
\label{theory:coverage_prob_NSP}
The coverage probability for an arbitrarily located user under a Nakagami fading serving channel while both shape parameter and rate parameter of gamma distribution\footnote{Channel gain, being the square of Nakagami random variable, follows a gamma distribution.} are $m_0$, is
\begin{align}
\label{eq:Pc}
\nonumber
\Pc\left(T\right)
&\triangleq \Prob\left(\SINR > T\right)\\\nonumber
&\hspace{-1.2 cm} = \int_{\rmin}^{\rmax}\int_{0}^{\infty} f_{\Xs}(\xs)f_{\Rs}\left(\rs\right)\Bigg[e^{-s\sigma^2}\\
&\hspace{-1.2 cm} \sum_{k=0}^{m_0-1}\hspace{-0 mm}\frac{\sum_{l=0}^{k}\binom{k}{l}\hspace{-1 mm}\left(s\sigma^2\right)^l \hspace{-1 mm}\left(-s\right)^{k-l}\hspace{-1 mm}\frac{\partial^{k-l}}{\partial s^{k-l}}\mathcal{L}_{\Insp}(s)}{k!}\Bigg]_{s=\frac{m_0T\rs^\alpha}{\ps \xs }}\,d\xs d\rs,
\end{align}
where the PDF $f_{\Rs}\left(\rs\right)$ is given in Lemma~\ref{Lem:pdf_r0} (or Lemma~\ref{Lem:r0 uniform}), $f_{\Xs}(\xs)$ is the PDF of $\Xs$ and $\mathcal{L}_{\Insp}\left(s\right)$ is the Laplace transform of interference power $\Insp$ calculated in the next section.

\end{theorem}
\begin{proof}
See Appendix~\ref{Appen:proof Theorem 1}.
\end{proof}

Let us then move on the average achievable data rate (in bit/s/Hz), which is {the ergodic capacity for a fading communication link derived from Shannon-Hartley theorem normalized to unit bandwidth.}
The average achievable rate is defined as
\begin{align}
\label{eq:data rate def.}
\bar{C} \triangleq\frac{1}{K}~\mathbb{E}\left[\log_2\left(1+\SINR\right)\right].
\end{align}
Unlike for the coverage probability, {frequency reuse} affects the average rate in two opposite directions. One direction is the improvement in SINR due to the reduction in the number of interfering satellites which use the same channel. The other direction which results in lower data rate is induced by a reduction in the availability of the frequency band shared among a group of satellites.

In the following theorem, we calculate the expression for the average rate of a user over Nakagami fading serving channel. The interfering channel gains may follow any arbitrary distribution.

\label{Appen:proof of Theory 3}
\begin{theorem}
\label{theorem:data_rate NSP}
The average data rate of an arbitrarily located user under a Nakagami fading serving channel and any fading or shadowing distribution for interfering channels is given by
\begin{align}
\nonumber
\bar{C}&= \frac{1}{ K}\hspace{-1 mm}\int_{\rmin}^{\rmax}\hspace{-2 mm}\int_{0}^{\infty}\hspace{-2 mm}\int_0^{\infty} f_{\Xs}(\xs)f_{\Rs}\left(\rs\right)\Bigg[e^{-s\sigma^2}\\
&\hspace{-4 mm}\sum_{k=0}^{m_0-1}\hspace{-0 mm}\frac{\sum_{l=0}^{k}\binom{k}{l}\hspace{-1 mm}\left(s\sigma^2\right)^l \hspace{-1 mm}\left(-s\right)^{k-l}\hspace{-1 mm}\frac{\partial^{k-l}}{\partial s^{k-l}}\mathcal{L}_{\Insp}(s)}{k!}\Bigg]_{s=\frac{m_0\left(2^t-1\right)\rs^\alpha}{\ps \xs }}\hspace{-1 cm}dt d\xs d\rs,
\end{align}
where $m_0$ is the parameter of Nakagami fading, and $\mathcal{L}_{\Insp}(s)$ will be given in Lemma~\ref{lem:Laplace NSP} and its corresponding corollaries, which cover some special cases.
\end{theorem}
\begin{proof}
See Appendix~\ref{Appen:proof of Theory 2}.
\end{proof}

\subsection{Interference Analysis}
\label{subsec:interference}
In this subsection, we derive the Laplace function of interference which is a key element of the performance expressions in Theorems~\ref{theory:coverage_prob_NSP} and \ref{theorem:data_rate NSP}. We, first, obtain the expression considering a general propagation model which means that no assumption is made regarding the specific expressions of $f_{\Xn}(x_n)$ and $f_{G_n}(g_n)$. 
\begin{lem}
\label{lem:Laplace NSP}
When the server is at distance $\rs \geq \rmin$ from the user and interfering channels experience an arbitrary distributed fading, the Laplace transform of random variable $\Insp$ is 
\begin{align}
\label{eq:Laplace NSP}
\mathcal{L}_{\Insp}(s) &\triangleq\mathbb{E}_{I}\left[e^{-s\Insp}\right]=\sum_{n=0}^{\infty}P_{n}\left(\mathcal{A}\left(\rmax\right)-\mathcal{A}\left(\rs\right)\right)\\
\nonumber
&\hspace{-1.4 cm}\times\hspace{-1 mm}\Bigg(\hspace{-0.5 mm}\int_{\rs}^{\rmax}\hspace{-2 mm}\int_0^{\infty}\hspace{-2 mm}\mathcal{L}_{G_n}(s\ppi x_n r_n^{-\alpha})f_{\Xn}(x_n)  f_{R_n|\Rs}(r_n|\rs)d\xs dr_n\hspace{-1 mm}\Bigg)^n,
\end{align}
where 
\begin{equation}
    \label{interfering_dist}
    f_{R_n|\Rs}(r_n|\rs)=\frac{{d\gamma(r_n)}/{d r_n}}{\gamma(\rmax)-\gamma(\rs)}
\end{equation}
is the probability density function of the distance from any visible satellite to the user conditioned on the serving distance \cite[Lemma~3]{okati2020downlink}. The parameter $\mathcal{A}\left(\rmax\right)$ represents the spherical cap where all satellites that can be viewed by the user exist while $\mathcal{A}\left(\rs\right)$ is the cap above the user, empty of satellites and with the serving satellite on its border (base of the cap). The function $f_{\Xn}(x_n)$ denotes the shadowing PDF for the $n$th interfering channel.
\end{lem}
\begin{proof}
See Appendix \ref{Appen:proof Lemma 5}.
\end{proof}

In the special cases, where channels experience Nakagami fading without any shadowing, Lemma~\ref{lem:Laplace NSP} can be reduced into the following corollary. The expression thereof is obtained by substituting the Laplace function of a gamma random variable, i.e., $\mathcal{L}_{G_n}(z)=\frac{\mn^{\mn}}{(\mn+z)^{\mn}}$, where $m_n$ stands for both shape parameter and rate parameter of gamma distribution for the $n$th link.

\begin{cor}
When the interfering channels experience only Nakagami fading (no shadowing), the Laplace function of interference can be written as 
\label{cor:LapNakagami}
\begin{align}
\label{eq:Laplace NSP no shadow}
\nonumber
\mathcal{L}_{\Insp}(s)&=\sum_{n=0}^{\infty}P_{n}\left(\mathcal{A}\left(\rmax\right)-\mathcal{A}\left(\rs\right)\right)\\
&\hspace{-1 cm}\times\Bigg(\int_{\rs}^{\rmax}\hspace{-2 mm}\frac{\mn^{\mn}}{(\mn+s\ppi r_n^{-\alpha})^{\mn}}  f_{R_n|\Rs}(r_n|\rs) dr_n\Bigg)^n,
\end{align}
where $\mathcal{A}\left(\rmax\right)$ and $\mathcal{A}\left(\rs\right)$ are the visible cap and the null cap above the user, respectively. The PDF $f_{R_n|\Rs}(r_n|\rs)$ is given in \eqref{interfering_dist}, and $m_n$ is the Nakagami fading parameter for $n$th link.
\end{cor}

When the intensity of the PPP is presumed to be constant (regardless of the latitude), the Laplace function can be obtained from the following corollary by simply substituting $\gamma(\cdot)$ in \eqref{interfering_dist} by the product of the density in Lemma~\ref{lem:constant Intensity} and the surface area of the spherical cap formed by the distance between the user and the given interfering satellite.
\begin{cor}
\label{cor:LapUniform}
The Laplace function of interference when the satellites are distributed uniformly with constant intensity, and their channels experience Nakagami fading, is given by 
\begin{align}
\label{eq:Laplace NSP const}
\nonumber
&\mathcal{L}_{\Insp}(s)\\\nonumber
&=\sum_{n=0}^{\infty}\frac{1}{n!}\left(\frac{N \left(\rmax^2-\rs^2\right)}{4\rEarth\left(\rEarth+\rmin\right)}\right)^n  \exp \left(- \frac{N \left(\rmax^2-\rs^2\right)}{4\rEarth\left(\rEarth+\rmin\right)}\right) \\
&\times\Bigg(\int_{\rs}^{\rmax}\hspace{-2 mm}\int_0^{\infty}\hspace{-2 mm}\frac{2r_n\mn^{\mn}f_{\Xs}(\xs)}{\left(\rmax^2-\rs^2\right)(\mn+s\ppi r_n^{-\alpha})^{\mn}}  \frac{}{}d\xs dr_n\Bigg)^n,
\end{align}
where  $f_{\Xn}(x_n)$ is the PDF of the shadowing component and $m_n$ is the fading parameter for Nakagami fading.
\end{cor}

The following corollary presents the Laplace function of interference in closed-form, under the assumptions given in Corollary~\ref{cor:LapUniform} and additionally excluding shadowing from the propagation model.
\begin{cor}
\label{cor:cor1+cor2}
The Laplace function of interference, when the satellites are distributed uniformly with constant intensity and their channels experience only Nakagami fading (no shadowing), is given by 
\begin{align}
\nonumber
&\mathcal{L}_{\Insp}(s)\\\nonumber
&=\sum_{n=0}^{\infty}\frac{1}{n!}\left(\frac{N \left(\rmax^2-\rs^2\right)}{4\rEarth\left(\rEarth+\rmin\right)}\right)^n  \exp \left(- \frac{N \left(\rmax^2-\rs^2\right)}{4\rEarth\left(\rEarth+\rmin\right)}\right) \\\nonumber
&\hspace{0 cm}\times\frac{1}{\left(\rmax^2-\rs^2\right)}\Bigg[\rmax^2\,{}_2F_1\left(-\frac{2}{\alpha},m_n;\frac{\alpha-2}{\alpha};-\, \frac{s\ppi \rmax^{-\alpha}}{m_n}\right)^n\\
&\hspace{0 cm}-{\rs^2}\,{}_2F_1\left(-\frac{2}{\alpha},m_n;\frac{\alpha-2}{\alpha};-\, \frac{s\ppi \rs^{-\alpha}}{m_n}\right)^n\Bigg],
\end{align}
where $_2F_1\left(\cdot,\cdot;\cdot;\cdot\right)$ is the Gauss's hyper-geometric function and $m_n$ is the fading parameter.
\end{cor}

{Finally, using the function given in} \cite[Eq.~9.100]{table} and substitution from special parameter values, the above can be reduced to elementary functions. For instance, when $m=1$ and $\alpha=2$, we have
\begin{align}
\label{lap:Rayleigh}
\nonumber
&\mathcal{L}_{\Insp}(s)\\\nonumber
&=\sum_{n=0}^{\infty}\frac{1}{n!}\left(\frac{N \left(\rmax^2-\rs^2\right)}{4\rEarth\left(\rEarth+\rmin\right)}\right)^n  \exp \left(- \frac{N \left(\rmax^2-\rs^2\right)}{4\rEarth\left(\rEarth+\rmin\right)}\right) \\
&\hspace{1 cm}\times \left(1+\frac{s\ppi}{\left(\rmax^2-\rs^2\right)}\ln\left(\frac{k+\rs^2}{k+\rmax^2}\right)\right).
\end{align}

To perform frequency reuse, we assign each satellite randomly and independently to a particular frequency channel. Therefore, the satellites assigned to each of the orthogonal frequency channels form a thinned version of the original PPP with intensity $\delta(\phisat)/K$. Since thinning preserves the Poisson point process according the thinning theorem of PPP \cite{blaszczyszyn2018stochastic}, we can take into account the effect of frequency reuse by substitution $\delta(\phisat)\to \delta(\phisat)/K$ in Laplace function of interference (in Lemma~\ref{lem:Laplace NSP} or the corresponding corollaries). Since the same frequency channel is used by the user and its nearest satellite, the frequency reuse has no effect on the original value of intensity that is used to obtain the PDF of the distance from the user to the server in Lemma~\ref{Lem:pdf_r0}.

\section{Numerical Results}
\label{sec:Numerical Results}

In this section, we provide numerical results to study the effect of different network parameters on  coverage probability and average data rate using the analytical expressions obtained in Section~\ref{sec:cov}. The performance of the network is evaluated in terms of satellite altitude and the number of orthogonal frequency channels, which provides important guidelines into the satellite network design. Furthermore, our analytical derivations are all verified through Monte Carlo simulations.

{In the propagation model, we consider the large-scale attenuation with path loss exponent $\alpha=2$,} the small-scale Nakagami-$m$ fading with integer $m\in \{1,2,3\}$, and lognormal shadowing. The lognormal shadowing is represented as $\Xs=10^{X_0/10}$ such that $X_0$ has a normal distribution with mean $\mus=0$ and standard deviation $\devs=9$ decibels. Thus, the PDF of lognormal shadowing is
\begin{align}
f_{\Xs}(\xs) = \frac{10}{\ln(10)\sqrt{2\pi}\devs\xs}\exp\left({-\frac{1}{2}\left(\frac{10\log_{10}(\xs)-\mus}{\devs}\right)^2}\right).
\end{align}
The number of orthogonal channels is set to $K=10$ in all the numerical results unless otherwise stated. For the reference simulations, satellites are placed uniformly on orbits centered at Earth's center with radius $\rEarth+\rmin$.

\begin{figure}[t]
    \includegraphics[trim = 0mm 0mm 0mm 0mm, clip,width=0.5\textwidth]{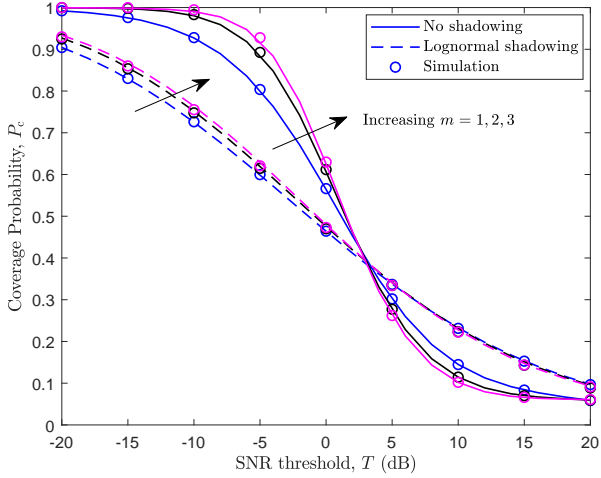}
    \caption{Theorem~\ref{theory:coverage_prob_NSP} verification with simulations when $\phiu=25^\circ$, $\iota=53^\circ$, $m\in\{1, 2, 3\}$, and $\thmin=10^\circ$.}
    \label{fig:CoverageVsT}
\end{figure}

Figure~\ref{fig:CoverageVsT} verifies our derivations given in Theorem~\ref{theory:coverage_prob_NSP} for $53^\circ$ inclined orbits and a user located at $25^\circ$ latitude. The total number of satellites and constellation altitude are chosen to be 2000 and 500~km, respectively. As shown in the figure, the markers that depict the Monte Carlo simulation results are completely matched with the lines that represent our theoretical expressions. As the chance of the user being in outage increases, shadowing affects the coverage probability more positively, the reason being that shadowing randomness increases the chance of a user with poor average $\SINR$ to be in coverage. Larger values of $m$ correspond to higher elevation angles and, consequently, less multi-path distortion, which result in slightly better coverage, but shadowing masks the effect of fading at large.

\begin{figure}[t]
    \includegraphics[trim = 0mm 0mm 0mm 0mm, clip,width=0.5\textwidth]{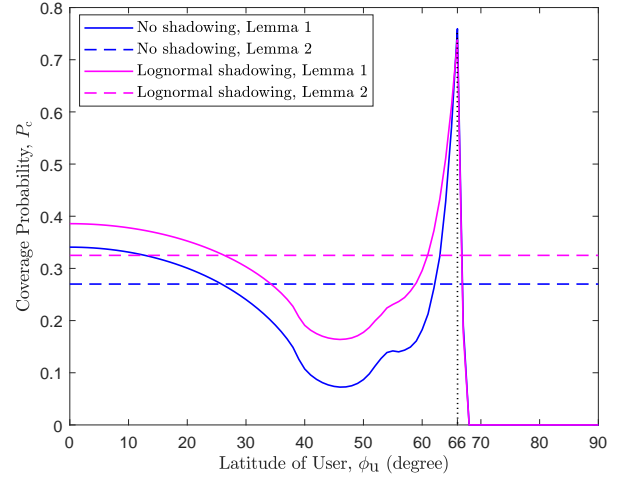}
    \caption{Theoretical coverage probability on different users' latitudes, $T=5$ dB, $\rmin=500$  km, $\iota=53^\circ$, $m=2$, and $\thmin=10^\circ$. The expression given in Theorem~\ref{theory:coverage_prob_NSP} is used to plot this figure.}
    \label{fig:CovVsPhiV1}
\end{figure}

\begin{figure}[t]
    \includegraphics[trim = 0mm 0mm 0mm 0mm, clip,width=0.5\textwidth]{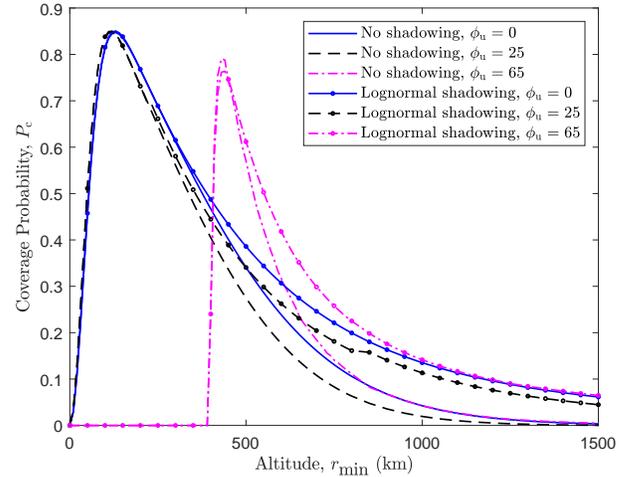}
    \caption{Effect of altitude on coverage probability when $T=5$ dB, $\phiu=\{0^\circ,25^\circ,65^\circ\}$, $\iota=53^\circ$, $m=2$, and $\thmin=10^\circ$. The expression given in Theorem~\ref{theory:coverage_prob_NSP} is used to plot this figure.}
    \label{fig:CoverageVsH}
\end{figure}

The effect of user's latitude on coverage probability is depicted in Fig.~\ref{fig:CovVsPhiV1}. The solid lines show the coverage probability when intensity of PPP is the function given in Lemma~\ref{lem:Intensity ppp} which forms a nonhomogeneous PPP with the intensity increasing when the user moves from the equator to the polar regions towards the inclination limits. As the user moves to higher latitudes, the performance becomes more unreliable due to increase in the density of satellites that use the same frequency channel as the user. As the spherical cap above the user (where all visible satellites to the user reside) leaves the constellation borders, the coverage probability starts increasing due to the reduction in the number of visible interferers. The coverage reaches its maximum at about $66^\circ$ where the only (if any) visible satellite to the user is the serving one, i.e., the performance becomes noise-limited. For latitudes larger than $66^\circ$, the coverage converges to zero quickly, since there are no satellites above horizon to serve the user. When the intensity of satellites is selected according to Lemma~\ref{lem:constant Intensity}, the coverage probability remains constant (dashed lines) all over the Earth for any latitude.

Figure~\ref{fig:CoverageVsH} {illustrates the probability of coverage at different latitudes of the user.} For all propagation environments, the coverage probability increases to reach its maximum value as the altitude increases which is then followed by a decline due to the rise in the number of visible interfering satellites. The optimum altitude for $\phiuser=0^\circ$ is slightly larger than $\phiuser=25^\circ$, the reason being that the intensity of satellites, and consequently the number of interferers, is higher at upper latitudes. When the user's latitude is set to $\phiuser=65^\circ$, which means that the user is located out of the constellation borders ($\phiu>\iota=53^\circ$), a larger altitude is crucial for the constellation so that the user can be served by a satellite within its visible range. As a result, for altitudes lower than about $400$~km, no visible satellite is available to serve the user which results in zero coverage probability.

\begin{figure}[t]
    \includegraphics[trim = 0mm 0mm 0mm 0mm, clip,width=0.5\textwidth]{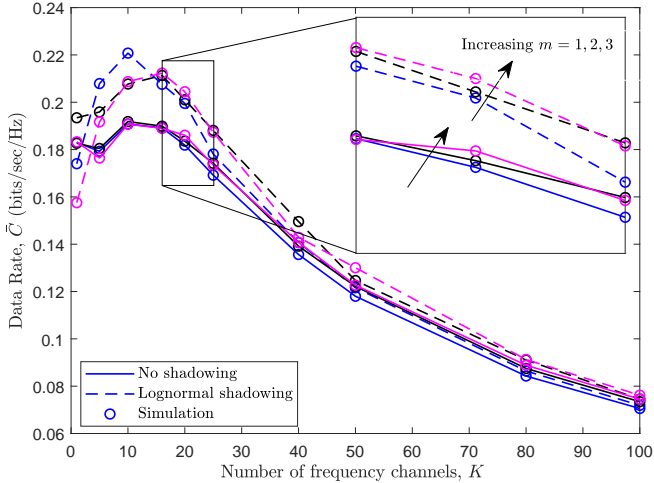}
    \caption{Theorem~\ref{theorem:data_rate NSP} verification with simulations when $\phiu=25^\circ$, $\iota=53^\circ$, $m\in\{1, 2, 3\}$, and $\thmin=10^\circ$.}
    \label{fig:DRVerification}
\end{figure}

Figure~\ref{fig:DRVerification} verifies the derivations for average data rate of a user at the latitude of $25^\circ$. As shown in the figure, the simulation results are perfectly in line with theoretical expressions in Theorem~\ref{theorem:data_rate NSP}. {The disparate behaviour of the curves is caused by the two opposite effects of frequency reuse on the average achievable rate.} As the number of frequency bands increases, the total number of interfering satellites on the same frequency band declines which results in an increase in the data rate. On the other hand, by increasing the number of frequency channels, the bandwidth shared among a group of satellites is reduced. 
{An increase in the plot is observed at first as a result of the decrease in interference received power, followed by a drop which is due to comprising only $\frac{1}{K}$ of frequency band.}

Figure~\ref{fig:DRVsPhi}, as a counterpart for Fig.~\ref{fig:CovVsPhiV1},  illustrates the variation of data rate over different latitudes when the intensity of Poisson point process is chosen to be according to Lemma~\ref{lem:Intensity ppp} or, for comparison, Lemma~\ref{lem:constant Intensity}. With intensity being as in Lemma~\ref{lem:Intensity ppp}, data rate varies over the different user's latitudes as shown in Fig.~\ref{fig:DRVsPhi}. For $53^\circ$ inclined orbits, there is a minor decline in data rate which is followed by a sharp rise due to a decrease in the number of visible interfering satellites when the user leaves the inclination limits.

\begin{figure}[t]
    \includegraphics[trim = 0mm 0mm 0mm 0mm, clip,width=0.5\textwidth]{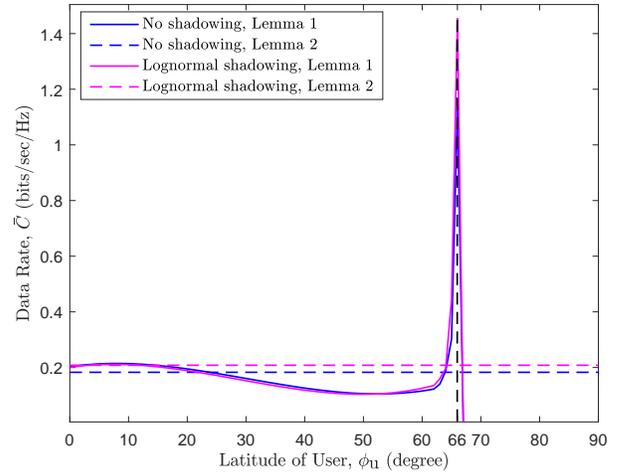}
    \caption{Data rate on different users' latitudes, $K=10$, $\rmin=500$~km, $\iota=53^\circ$, $m=2$, and $\thmin=10^\circ$. The expression given in Theorem~\ref{theorem:data_rate NSP} is used to plot this figure.}
    \label{fig:DRVsPhi}
\end{figure}
\begin{figure}[t]
    \includegraphics[trim = 0mm 0mm 0mm 0mm, clip,width=0.5\textwidth]{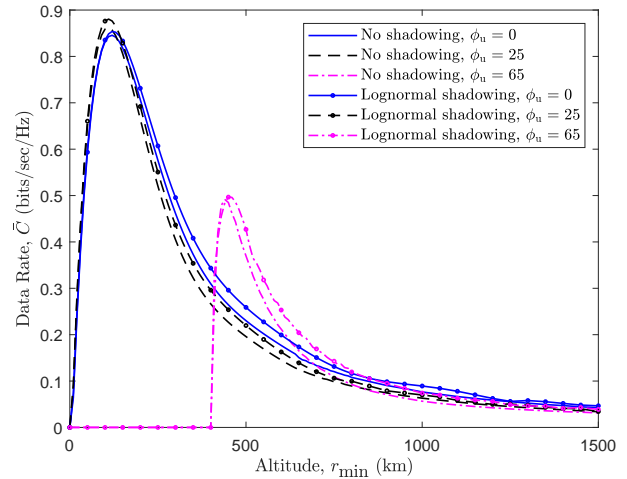}
    \caption{Effect of altitude on data rate when $K=10$, $\phiu=\{0^\circ,25^\circ,65^\circ\}$, $\iota=53^\circ$, $m=2$, and $\thmin=10^\circ$. The expression given in Theorem~\ref{theorem:data_rate NSP} is used to plot this figure.}
    \label{fig:DRVsH}
\end{figure}

The effect of altitude on data rate is depicted in Fig.~\ref{fig:DRVsH} for $K=10$. Similar to Fig.~\ref{fig:CoverageVsH}, from the minimum altitude at which the user is able to visit at least one satellite, the data rate increases rapidly until reaches a maximum point. After the maximum point, the data rate decreases more slowly due to the increase in the number of interfering satellites as well as satellites being at a farther distance from the user. The altitude which maximizes the data rate varies with the user's latitude and, obviously, it has similar value which results in the maximum coverage probability.

\section{Conclusions}
\label{sec:Conclusion}

In this paper, a generic approach to obtain the analytical performance of a massive low Earth orbit network was presented by modeling the satellites' locations as a nonhomogeneous Poisson point process and utilizing the tools from stochastic geometry. The density of the nonhomogeneous Poisson point process is derived from the actual geometry of the constellation which enables us to take into account the non-uniform distribution of satellites across different latitudes. {Our next step was to apply this model to derive analytical expressions for the coverage probability and average data rate of an arbitrarily located user in terms of} the distribution of fading and shadowing as well as the Laplace function of interference. From numerical results, we concluded that shadowing affects the coverage probability positively as the probability of the user being in outage increases. 
Furthermore, we showed how the analytical results allow one to find --- without involved orbital simulations --- optimum values for altitude, the number of orthogonal frequency channels, and user's latitude which result in {the largest coverage and/or throughput}, given constellation parameters.

\appendix
\subsection{Proof of Lemma~\ref{Lem:pdf_r0}}
\label{Appen:proof of Lemma 2}
For an NPPP, the CDF of $\Rs$ can be written as
\begin{align}
    \label{eq:CDF_R0}
   F_{\Rs}\left(\rs\right)=1-\Prob(\Rs>\rs)=1-\Prob(\mathcal{N}=0),
\end{align}
where $\Prob(\mathcal{N}=0)$ is the void probability of PPP in ${\cal A}(\rs)$ that can be obtained from \eqref{void prob} by setting $n=0$.
According to \eqref{void prob}, by integrating from the intensity over the spherical cap above the user, we have
\begin{align}
    \label{eq:CDF_R0_proof}
    \nonumber
   &F_{R_0}\left(r_0\right)\\\nonumber
   &=\hspace{-1 mm}1\hspace{-1 mm}-\exp\hspace{-1 mm}\Bigg(\hspace{-1 mm}-\hspace{-1 mm}\int_{\max(\phiu-\ths,-\iota)}^{\min(\phiu+\ths,\iota)}\hspace{-4 mm}\beta(\phisat)\delta(\phisat) \hspace{-0.5 mm}\left(\rmin+\rEarth\right)^2 \hspace{-0.5 mm}\cos(\phisat)d{\phisat}  \hspace{-1 mm}\Bigg)\\\nonumber
   &\stackrel{(a)}=\hspace{-1 mm}1\hspace{-1 mm}-\exp\hspace{-1 mm}\Bigg(\hspace{-1 mm}-2\left(\rmin+\rEarth\right)^2\int_{\max(\phiu-\ths,-\iota)}^{\min(\phiu+\ths,\iota)}\delta(\phisat)\cos(\phisat)\\
  & \hspace{3 cm}\times\cos^{-1}\left({\frac{\cos(\ths)}{\cos(\phisat-\phiu)}}\right) d\phisat \Bigg),
\end{align}
where $\beta(\phisat)$ is the longitude range inside the spherical cap above the user at latitude $\phisat$. Equality (a) follows from substitution of $\beta(\phisat)$ using the basic geometry. Taking the derivative of \eqref{eq:CDF_R0_proof} with respect to $\rs$ completes the proof of Lemma~\ref{Lem:pdf_r0}. Note that for $\ths\leq |{\phiu}| - \iota$ the CDF given in \eqref{eq:CDF_R0_proof} is zero since the spherical cap formed by polar angle $\ths$ above the latitude $\phiu$ is much farther from the constellation's borders to contain any satellite.

\subsection{Proof of Theorem~\ref{theory:coverage_prob_NSP}}
\label{Appen:proof Theorem 1}

To obtain the expression given \eqref{eq:Pc}, {let us begin with the definition of coverage probability:}
\begin{align}
\Pc\left(T\right)&=
\mathbb{E}_{\Rs}\left[\Prob\left(\SINR>T |\Rs\right)\right]\\\nonumber
&=\int_{\rmin}^{\rmax}\Prob\left(\SINR>T |\Rs=\rs\right)f_{\Rs}\left(\rs\right)d\rs\\\nonumber
&=\int_{\rmin}^{\rmax}\Prob\left(\frac{\ps \Gs \Xs \rs^{-\alpha}}{\Insp+\sigma^2}>T\right)f_{\Rs}\left(\rs\right)\,d\rs\\\nonumber
&=\int_{\rmin}^{\rmax}\mathbb{E}_{\Insp}\hspace{-0 mm}\left[\Prob\left(\hspace{-1 mm}\Gs\Xs>\frac{T\rs^\alpha\left(\Insp+\sigma^2\right)}{\ps }\hspace{-1 mm}\right)\bigg|\Insp>0\right]\hspace{-1 mm}\\\nonumber
&\hspace{1.5 cm}\times f_{\Rs}\left(\rs\right)d\rs.
\end{align}
{Since satellites with elevation angles smaller than $\thmin$ are not visible to the user, the integral has an upper limit.} Then
\begin{align}
\nonumber
&\Pc\left(T\right)=\\
\nonumber
&\stackrel{(a)}= \hspace{-1 mm}\int_{\rmin}^{\rmax}\mathbb{E}_{\Insp}\left[\int_{0}^{\infty} \hspace{-2 mm}f_{\Xs}(\xs)\hspace{-1 mm}\left(\hspace{-1 mm}1-F_{\Gs}\hspace{-1 mm}\left(\frac{T\rs^\alpha\left(\Insp+\sigma^2\right)}{\ps \xs }\right)\hspace{-1 mm}\right)\hspace{-1 mm}\right]\\\nonumber
&\hspace{1.5 cm}\times f_{\Rs}\left(\rs\right)\,d\xs d\rs\\\nonumber
&\stackrel{(b)}= \hspace{-1 mm}\int_{\rmin}^{\rmax}\mathbb{E}_{\Insp}\left[\int_{0}^{\infty}\hspace{-2 mm} f_{\Xs}(\xs)\hspace{-1 mm}\left(\frac{\Gamma\left(m_0,m_0\frac{T\rs^\alpha\left(\Insp+\sigma^2\right)}{\ps \xs }\right)}{\Gamma\left(m_0\right)}\right)\right]\\\nonumber
&\hspace{1.5 cm}\times f_{\Rs}\left(\rs\right)\,d\xs d\rs\\\nonumber
&\stackrel{(c)}= \int_{\rmin}^{\rmax}\hspace{-1 mm}\int_{0}^{\infty} \hspace{-1 mm}f_{\Xs}(\xs)f_{\Rs}\left(\rs\right)e^{-\frac{m_0T\rs^\alpha\sigma^2}{\ps \xs }}\mathbb{E}_{\Insp}\Bigg[e^{-\frac{m_0T\rs^\alpha \Insp}{\ps \xs }}\\\nonumber
&\sum_{k=0}^{m_0-1}\frac{\sum_{l=0}^{k}\binom{k}{l}\left(m_0\frac{T\rs^\alpha\sigma^2}{\ps \xs }\right)^l \left(m_0\frac{T\rs^\alpha \Insp}{\ps \xs }\right)^{k-l}}{k!}\Bigg]\,d\xs d\rs\\
\nonumber
&= \int_{\rmin}^{\rmax}\int_{0}^{\infty} f_{\Xs}(\xs)f_{\Rs}\left(\rs\right)\Bigg[e^{-s\sigma^2}\\
&\sum_{k=0}^{m_0-1}\hspace{-0 mm}\frac{\sum_{l=0}^{k}\binom{k}{l}\hspace{-1 mm}\left(s\sigma^2\right)^l \hspace{-1 mm}\left(-s\right)^{k-l}\hspace{-1 mm}\frac{\partial^{k-l}}{\partial s^{k-l}}\mathcal{L}_{\Insp}(s)}{k!}\Bigg]_{s=\frac{m_0T\rs^\alpha}{\ps \xs }}\,d\xs d\rs.
\end{align}
The substitution in (a) follows from the product distribution of two independent random variables, (b) follows from the distribution of gamma random variable $\Gs$ (being the square of the Nakagami random variable), and (c) is calculated by applying the incomplete gamma function for integer values of $m_0$ to (b).


\subsection{Proof of Theorem~\ref{theorem:data_rate NSP}}
\label{Appen:proof of Theory 2}
Most of the steps in derivation of the data rate are similar to those given in Appendix B. According to the definition of the average data rate given in \eqref{eq:data rate def.}, we have
\begin{align}
\nonumber
&\bar{C}=\mathbb{E}_{\Insp,\Gs,\Xs,\Rs}\left[\log_2\left(1+\SINR\right)\right]\\\nonumber
&=\hspace{-1 mm}\int_{\rmin}^{\rmax}\hspace{-1 mm}\mathbb{E}_{\Insp,\Gs,\Xs}\left[\log_2\left(1+\frac{\ps \Gs\Xs \rs^{-\alpha}}{\sigma^2+\Insp}\right)\right]f_{\Rs}(\rs)\,d\rs\\\nonumber
&\stackrel{}=\hspace{-1 mm}\int_{\rmin}^{\rmax}\hspace{-2 mm}\int_0^{\infty}\hspace{-1 mm}\mathbb{E}_{\Insp,\Gs,\Xs}\left[\Prob\left(\log_2\hspace{-1 mm}\left(\hspace{-1 mm}1+\frac{\ps \Gs\Xs \rs^{-\alpha}}{\sigma^2+\Insp}\hspace{-1 mm}\right)\hspace{-1 mm}>\hspace{-1 mm} t\right)\right]\hspace{-1 mm}\\
&\hspace{3 cm}\times f_{\Rs}(\rs)\,dt d\rs,
\end{align}
where the last step follows form the fact that for a positive random variable $X$, $\mathbb{E}\left[X\right]=\int_{t>0}\Prob\left(X>t\right)dt$. Thus, we have
\begin{align}
\nonumber
&\bar{C}=\hspace{-1 mm}\int_{\rmin}^{\rmax}\hspace{-2 mm}\int_0^{\infty}\mathbb{E}_{\Insp,\Gs,\Xs}\left[\Prob\left(\Gs\Xs>\frac{\rs^\alpha\left(\sigma^2+\Insp\right)}{\ps }\hspace{-1 mm}\left(2^t-1\right)\right)\right]\hspace{-1 mm}\\\nonumber
&\hspace{3 cm}\times f_{\Rs}(\rs)\,dt d\rs\\\nonumber
&\stackrel{(a)}=\hspace{-1 mm}\int_{\rmin}^{\rmax}\hspace{-2 mm}\int_0^{\infty}\hspace{-2 mm}\mathbb{E}_{\Insp}\hspace{-1 mm}\Bigg[\int_{0}^{\infty} \hspace{-1 mm}f_{\Xs}(\xs)\hspace{-1 mm}\\\nonumber
&\times\left(\hspace{-1 mm}1-F_{\Gs}\hspace{-1 mm}\left(\frac{\rs^\alpha\left(\Insp+\sigma^2\right)\left(2^t-1\right)}{\ps \xs }\right)\right)\hspace{-1 mm}d\xs\Bigg] f_{\Rs}(\rs)\, dt d\rs\\\nonumber
&\stackrel{(b)}=\int_{\rmin}^{\rmax}\hspace{-2 mm}\int_0^{\infty}\hspace{-2 mm}\mathbb{E}_{\Insp}\Bigg[\int_{0}^{\infty}\hspace{-1 mm} f_{\Xs}(\xs)\\\nonumber
&\times\left(\frac{\Gamma\left(m_0,m_0\frac{\rs^\alpha\left(\Insp+\sigma^2\right)\left(2^t-1\right)}{\ps \xs }\right)}{\Gamma\left(m_0\right)}\right)d\xs\Bigg] f_{\Rs}(\rs)\,dt d\rs\\\nonumber
&\stackrel{(c)}=\hspace{-1 mm}\int_{\rmin}^{\rmax}\hspace{-2 mm}\int_0^{\infty}\hspace{-2 mm}\int_{0}^{\infty} f_{\Xs}(\xs)f_{\Rs}(\rs) e^{-\frac{m_0\left(2^t-1\right)\rs^\alpha\sigma^2}{\ps \xs }}\\\nonumber
&\hspace{ 0 mm}\times \mathbb{E}_{\Insp}\Bigg[e^{-\frac{m_0\left(2^t-1\right)\rs^\alpha \Insp}{\ps \xs }}\\\nonumber
&\sum_{k=0}^{m_0-1}\hspace{-1 mm}\frac{\sum_{l=0}^{k}\hspace{-1 mm}\binom{k}{l}\hspace{-1 mm}\left(\frac{m_0\left(2^t-1\right)\rs^\alpha\sigma^2}{\ps \xs }\right)^l \hspace{-1 mm} \left(\frac{m_0\left(2^t-1\right)\rs^\alpha \Insp}{\ps \xs }\right)^{k-l}\hspace{-2 mm}}{k!}\Bigg]dt d\xs d\rs\\\nonumber
&= \hspace{-1 mm}\int_{\rmin}^{\rmax}\hspace{-2 mm}\int_{0}^{\infty}\hspace{-2 mm}\int_0^{\infty} f_{\Xs}(\xs)f_{\Rs}\left(\rs\right)\Bigg[e^{-s\sigma^2}\\
&\sum_{k=0}^{m_0-1}\hspace{-0 mm}\frac{\sum_{l=0}^{k}\binom{k}{l}\hspace{-1 mm}\left(s\sigma^2\right)^l \hspace{-1 mm}\left(-s\right)^{k-l}\hspace{-1 mm}\frac{\partial^{k-l}}{\partial s^{k-l}}\mathcal{L}_{I}(s)}{k!}\Bigg]_{s=\frac{m_0\left(2^t-1\right)\rs^\alpha}{\ps \xs }}\hspace{-2 mm}dt d\xs d\rs.
\end{align}
Similar to the proof of Theorem~\ref{theory:coverage_prob_NSP}, (a) follows from the product distribution of two independent random variables, (b) follows from the gamma distribution of serving channel gain $\Gs$, and (c) {is calculated by applying the incomplete gamma function for integer values of $m_0$ to (b).}

\subsection{Proof of Lemma~\ref{lem:Laplace NSP}}
\label{Appen:proof Lemma 5}
In this appendix, we derive the expression for Laplace function of interference assuming arbitrary distributions for fading and shadowing. Let us start with the definition of Laplace function for random variable $\Insp$ which is
\begin{align}
\label{eq:Laplace NSP proof}
\nonumber
&\mathcal{L}_{\Insp}(s) \triangleq\mathbb{E}_{\Insp}\left[e^{-s\Insp}\right]\\\nonumber
&=\mathbb{E}_{\mathcal{N}, R_n,\Xn, G_n}\left[\exp{\left(-s\sum_{n\in \xi \backslash \{\mathrm{s}\}}^{}\ppi G_n\Xn R_n^{-\alpha}\right)}\right]\\\nonumber
&=\mathbb{E}_{\mathcal{N}, R_n,\Xn, G_n}\left[\prod_{n\in \xi \backslash \{\mathrm{s}\}}^{}\exp{\left(-s\ppi G_n \Xn R_n^{-\alpha}\right)}\right]\\\nonumber
&\stackrel{(a)}= \mathbb{E}_{\mathcal{N}, R_n}\left[\prod_{n\in \xi \backslash \{\mathrm{s}\}}^{}\mathbb{E}_{\Xn, G_n}\left[\exp{\left(-s\ppi G_n \Xn R_n^{-\alpha}\right)}\right]\right]\\\nonumber
&\stackrel{(b)}= \mathbb{E}_{\mathcal{N}}\Bigg[\prod_{n\in \xi \backslash \{\mathrm{s}\}}^{}\int_{\rs}^{\rmax}\mathbb{E}_{\Xn, G_n}\left[\exp{\left(-s\ppi G_n \Xn r_n^{-\alpha}\right)}\right]\\
&\hspace{3 cm}\times f_{R_n|\Rs}(r_n|\rs)dr_n\Bigg],\\\nonumber
\end{align}
where  (a) follows from the i.i.d.\ distribution of $G_n$ and $\Xn$ as well as their independence from $\mathcal{N}$ and $R_n$ and (b) is obtained by taking the expectation over the random variable $R_n$ conditioned on $\Rs$. Then
\begin{align}
\nonumber
&\mathcal{L}_{\Insp}(s)\stackrel{(c)}= \sum_{n=0}^{\infty}P_{n}\left(\mathcal{A}\left(\rmax\right)-\mathcal{A}\left(\rs\right)\right)\\\nonumber
&\hspace{1 mm}\times\Bigg(\int_{\rs}^{\rmax}\mathbb{E}_{\Xn, G_n}\left[\exp{\left(-s\ppi G_n \Xn r_n^{-\alpha}\right)}\right] f_{R_n|\Rs}(r_n|\rs)dr_n\Bigg)^n\\\nonumber
&\stackrel{(d)}= \sum_{n=0}^{\infty}P_{n}\left(\mathcal{A}\left(\rmax\right)-\mathcal{A}\left(\rs\right)\right)\\
&\hspace{0 cm}\times\Bigg(\int_{\rs}^{\rmax}\int_0^{\infty}\hspace{-2 mm}\mathcal{L}_{G_n}(s\ppi x_n r_n^{-\alpha})f_{\Xn}(x_n)  f_{R_n|\Rs}(r_n|\rs)d\xs dr_n\Bigg)^n,
\end{align}
where $\mathcal{A}\left(\rmax\right)-\mathcal{A}\left(\rs\right)$ indicates the region above the user where satellites which are more distanced from the user than the serving satellite exist, (c) is obtained by taking the average over the Poisson random variable $\mathcal{N}$, and applying the law of total expectation on independent random variables $G_n$ and $\Xn$ results in (d).


\IEEEtriggeratref{24}
\bibliography{ref} 
\bibliographystyle{IEEEtran}



\end{document}